\documentclass[a4paper]{article}
\usepackage{amsmath,pdfsync}
\usepackage{amsthm}
\usepackage{amsfonts}
\usepackage{amssymb}

\theoremstyle{plain}
\newtheorem{thm}{Theorem}[section]

\newtheorem{prop}[thm]{Proposition}
\newtheorem{lemma}[thm]{Lemma}
\newtheorem{cor}[thm]{Corollary}
\theoremstyle{definition}

\newtheorem{alg}{Algorithm}
\theoremstyle{remark}
\newtheorem{remark}[thm]{Remark}

\begin{document}
\title{Simulating Events of Unknown Probabilities via Reverse Time Martingales}

\author{Krzysztof \L atuszy\'nski \\  \textit{Department of Statistics} \\ \textit{University of Warwick} \\ \small{Coventry, CV4 7AL} \\ \small{latuch@gmail.com} \and Ioannis Kosmidis \\ \textit{Department of Statistics} \\ \textit{University of Warwick} \\  \small{Coventry, CV4 7AL} \\ \small{I.Kosmidis@warwick.ac.uk} \and Omiros Papaspiliopoulos \\ \textit{Department of Economics} \\ \textit{Universitat Pompeu Fabra} \\ \small{Ramon Trias Fargas 25-27} \\ \small{Barcelona 08005, Spain} \\ \small{omiros.papaspiliopoulos@upf.edu}  \and Gareth O. Roberts \\ \textit{University of Warwick} \\ \textit{Department of Statistics} \\ \small{Coventry, CV4 7AL} \\ \small{gareth.o.roberts@warwick.ac.uk}}
\maketitle
\begin{abstract}
Let $s\in
(0,1)$ be uniquely determined but only its approximations
can be obtained with a finite computational effort. Assume one aims to simulate an event of probability $s.$ Such settings are
often encountered in statistical simulations. We consider two specific
examples. First, the exact simulation of non-linear diffusions
(\cite{BeskosRobertsEA1}). Second, 
the celebrated Bernoulli factory problem (\cite{KeaneOBrien},
\cite{NacuPeres}) of generating an $f(p)-$coin given a sequence $X_1,
X_2,...$ of 
independent tosses of a $p-$coin (with known $f$ and unknown $p$). We
describe a general framework and provide algorithms where this kind of
problems can be fitted and solved. The algorithms are straightforward
to implement and thus allow for effective simulation of desired events
of probability $s.$ 
Our methodology  links the simulation problem to existence and
construction of unbiased estimators. 

\end{abstract}
\section{Introduction} \label{sec_intro}
Assume that one aims to simulate an event of unknown probability $s\in
(0,1)$ which is uniquely determined, however only its approximations
can be obtained using a finite computational effort. Such settings are often encountered in
statistical simulations and emerge if e.g. $s$ is given by a series expansion or a consistent estimator for $s$ is available (see e.g. \cite{Devroye},
\cite{BeskPapaRobFearn}, \cite{BeskosPapaRobertsEA3}, \cite{BeskosRobertsEA1}
\cite{KeaneOBrien}, \cite{NacuPeres}). A celebrated example of this
kind is the Bernoulli factory problem which motivated our work. It can
be stated as follows. Let $p \in \mathcal{P} \subseteq [0, 1]$ be
unknown and let $f : \mathcal{P} \to [0, 1].$ Then the problem is to generate $Y,$ a
single coin toss of an $s=f(p)-$coin, given a sequence $X_1, X_2,...$
of independent tosses of a $p-$coin. For the historical context of this
question and 
a range of theoretical results see \cite{PeresAnnStat}, \cite{KeaneOBrien}, \cite{NacuPeres}, \cite{MosselPeres} and \cite{HoltzNazarovPeres}.
In particular \cite{KeaneOBrien} provide necessary and sufficient conditions
for $f,$ under which an algorithm generating an $f(p)-$coin
exists. Nacu and Peres in \cite{NacuPeres} suggest a constructive
algorithm for simulating $f(p) = \min\{2p, 1- 2\varepsilon\}$ which is
central to solving the problem 
for general $f$ and allows for generating an $f(p)-$coin for a large
class of functions (e.g. real analytic, see \cite{NacuPeres} and
Section \ref{sec_NacuPeres} for
details).
The algorithm is based on polynomial envelopes of $f.$ To
run the algorithm one has to construct sets of $\{0,1\}$ strings of
appropriate cardinality based on coefficients of the polynomial
envelopes. Unfortunately its naive implementation requires dealing
with sets of exponential size (we encountered e.g. $2^{2^{26}}$) and thus is not very
practical. Hence the authors provide a simple approximate algorithm
for generating $\min\{2p,1\}-$coins.
 Ongoing research in Markov chain Monte Carlo and
rejection sampling indicates that the Bernoulli factory problem is not
only of theoretical interest (c.f. \cite{AssGlynnThor},
\cite{HendeGlynn},  Chapter~16 of \cite{AssGlynnBook}, and Section~\ref{sec:exact} of the present paper). However using
approximate algorithms in these applications perturbs simulations in a
way difficult to quantify.

In Section 2   we develop a framework for simulating events of
unknown probabilities. Our approach is based on random sequences, say
$L_n$ and $U_n$ under- and overestimating $s$ that are monotone in
expectations (i.e. $\mathbb{E}\;L_n \nearrow s$ and $\mathbb{E}\;U_n
\searrow s$) and are reverse time super- and submartingales
respectively. From $L_n$ and $U_n$ we construct $\tilde{L}_n$ and
$\tilde{U}_n$ that are monotone almost surely and have the same
expectations (c.f. Theorem \ref{thm_alg_4_valid}, Algorithm 4). Given
$\tilde{L}_n$ and $\tilde{U}_n$ we sample events of probability $s$
using a single $U(0,1)$ random variable.  This result generalizes
classical constructions for simulation of events of unknown
probabilities using  deterministic sequences (\cite{Devroye}). We link these results to existence and construction of unbiased estimators. In particular one can use the algorithms of Section \ref{sec_theory} to obtain unbiased sequential estimators of a parameter of interest that is not necessarily in $[0,1].$

We illustrate our results with examples.  First, in Section
\ref{sec_NacuPeres}, we present a reverse time martingale/unbiased
estimator formulation of the Nacu-Peres algorithm which we believe
gives a new perspective on the Bernoulli factory problem. We identify
the coefficients of the lower and
 upper polynomial envelopes as random variables of desired properties
 and implement the algorithm using a single $U(0,1)$ auxiliary random
 variable. We do not need to identify subsets of $\{0,1\}$ strings and
 thus avoid algorithmic difficulties of the original version. The martingale approach also simplifies the proof of validity of the Nacu-Peres algorithm. In the special case when $f$ has an alternating series expansion with decreasing coefficients, the martingale approach results in new, very efficient algorithms. Second, in Section
\ref{sec:exact} we obtain the Exact 
Algorithm for diffusions introduced in \cite{BeskosRobertsEA1} as an application of the generic
Algorithm 3 of Section \ref{sec_theory}. 

\section{Simulation of Events with Unknown Probabilities}
\label{sec_theory}
Throughout the paper we assume that we can generate uniformly distributed
iid random variables $G_0, G_1,... \sim U(0, 1)$ which will serve as a source of randomness for algorithms. Thus to simulate an $s-$coin $C_s$ we just let $C_s := \mathbb{I}\{G_0 \leq s\}.$ We will be concerned with settings where $s$ is not known explicitly.

The following simple observation will turn out very useful. 
\begin{lemma} \label{lemma_alg_1_valid}
Sampling events of probability $s \in [0,1]$ is equivalent to constructing an unbiased estimator of $s$ taking values in $[0,1]$ with probability~1.    
\end{lemma}
\begin{proof}
Let $\hat{S},$ s.t. $\mathbb{E} \hat{S} = s$ and $\mathbb{P}(\hat{S} \in [0,1]) = 1$ be the estimator. Then draw $G_0 \sim U(0,1),$ obtain $\hat{S}$ and define a coin $C_s := \mathbb{I}\{G_0 \leq \hat{S}\}.$ Clearly $$\mathbb{P}(C_s = 1) = \mathbb{E}\;\mathbb{I}(G_0 \leq \hat{S}) =  \mathbb{E}\left(\mathbb{E} \left( \mathbb{I}(G_0 \leq \hat{s})\;|\; \hat{S} = \hat{s}\right)\right) = \mathbb{E}\hat{S} = s.$$
The converse is straightforward since an $s-$coin is an unbiased estimator of $s$ with values in $[0,1].$
\end{proof}
Thus given $\hat{S} \in [0,1],$ an unbiased estimator of $s,$ we can sample events of probability $s$ by the following algorithm.  
\begin{alg} \label{alg_1} ~\\
1. simulate $G_0 \sim U(0,1);$\\
2. obtain $\hat{S};$\\
3. if $G_0 \leq \hat{S}$ set $C_s:=1,$ otherwise set $C_s:=0;$\\
4. output $C_s.$
\end{alg}
Next assume that $l_1, l_2,...$ and $u_1, u_2,...$ are sequences of lower and upper bounds for $s$ converging to $s.$ This setting is well known (\cite{Devroye}) and appears in a variety of situations, usually as an element of more complex simulation procedures, see e.g. \cite{BeskosPapaRobertsEA3}, \cite{PapaRobe}. Here we use the following algorithm for simulating an $s-$coin.
\begin{alg} \label{alg_2} ~\\
1. simulate  $G_0 \sim U(0,1);$ set $n = 1;$\\
2. compute $l_n$ and $u_n;$\\
3. if $G_0 \leq l_n$ set $C_s := 1;$\\
4. if $G_0 > u_n$ set $C_s := 0;$\\
5. if $l_n < G_0 \leq u_n$ set $n := n + 1$ and GOTO 2;\\
6. output $C_s.$
\end{alg}
The algorithm stops with probability $1$ since $l_n$ and $u_n$ converge to $s$ from below and from above. Precisely, the algorithm needs $N > n$ iterations to stop with probability $\inf_{k\leq n} u_k - \sup_{k\leq n} l_k.$ Because we can always obtain monotone bounds by setting $u_n := \inf_{k\leq n} u_k$ and $l_n := \sup_{k\leq n} l_k,$ we assume that $l_n$ is an increasing
sequence and $u_n$ is a decreasing sequence.

The next step is to combine the above ideas and work with randomized bounds, i.e. in a setting where we have estimators $L_n$ and $U_n$ of the upper and lower bounds $l_n$ and $u_n.$ The estimators shall live
on the same probability space and have the following properties that hold a.s. for every $n = 1,2, ...$
\begin{eqnarray}
& & L_n \leq U_n  \label{assu_L_U_ineq} \\
& & L_n \in [0, 1] \qquad \textrm{and}\qquad U_n \in [0, 1] \label{assu_L_U_in_01} \\
& & L_{n-1} \leq L_n \quad \; \; \,\textrm{and} \qquad U_{n-1} \geq U_n \qquad \qquad \label{assu_L_U_monotone}
\end{eqnarray}
Note that we do not assume that $L_n \leq s$ or $U_n \geq s.$ Also condition
(\ref{assu_L_U_monotone}) implies monotonicity of expectations, i.e.
\begin{eqnarray} \mathbb{E}\;L_n = l_n \nearrow s &\textrm{and}& \mathbb{E}\;U_n = u_n \searrow s. \label{assu_L_U_Expect} \end{eqnarray}
Let $$\mathcal{F}_0 = \{\emptyset, \Omega\}, \qquad \mathcal{F}_n= \sigma\{L_n,U_n\}, \qquad \mathcal{F}_{k,n} = \sigma\{\mathcal{F}_k, \mathcal{F}_{k+1},... \mathcal{F}_n\} \quad \textrm{for $k\leq n.$} $$ Consider the following algorithm.
\begin{alg} \label{alg_3} ~\\
1. simulate  $G_0 \sim U(0,1);$ set $n = 1;$\\
2. obtain $L_n$ and $U_n$ given $\mathcal{F}_{0, n-1},$\\
3. if $G_0 \leq L_n$ set $C_s := 1;$\\
4. if $G_0 > U_n$ set $C_s := 0;$\\
5. if $L_n < G_0 \leq U_n$ set $n := n + 1$ and GOTO 2;\\
6. output $C_s.$
\end{alg}
\begin{lemma} \label{lemma_alg_3_valid}
Assume (\ref{assu_L_U_ineq}), (\ref{assu_L_U_in_01}), (\ref{assu_L_U_monotone}) and (\ref{assu_L_U_Expect}). Then Algorithm \ref{alg_3} outputs
a valid $s-$coin. Moreover the probability that it needs $N>n$ iterations equals $u_n-l_n.$ \end{lemma}
\begin{proof} Probability that Algorithm \ref{alg_3} needs more then $n$ iterations equals $\mathbb{E}(U_n - L_n) = l_n - u_n \to 0$ as $n \to \infty.$ And since $0 \leq U_n - L_n$ is a decreasing sequence a.s., we also have $U_n - L_n \to 0$ a.s. So there exists a random variable $\hat{S},$ such that for
almost every realization of sequences $\{L_n(\omega)\}_{n\geq 1}$ and $\{U_n(\omega)\}_{n\geq 1}$ we have $L_n(\omega) \nearrow \hat{S}(\omega)$
and $U_n(\omega) \searrow \hat{S}(\omega).$ By (\ref{assu_L_U_in_01}) we have $\hat{S} \in [0,1]$ a.s. Thus for a fixed $\omega$ the algorithm outputs an
$\hat{S}(\omega)-$coin a.s. Clearly $\mathbb{E}\; L_n \leq \mathbb{E} \; \hat{S} \leq \mathbb{E} \; U_n$ and hence $\mathbb{E}\;\hat{S} = s.$
\end{proof}
\begin{remark}
The random variable $\hat{S}$ constructed in the proof can be viewed as the unbiased estimator of $s$ mentioned earlier with sequences $L_n$ and $U_n$ being its lower and upper random approximations.
\end{remark} 
\begin{remark}
For Algorithm \ref{alg_3} assumption (\ref{assu_L_U_in_01}) can be relaxed to \begin{eqnarray}
L_n \in (-\infty, 1] &\textrm{and}& U_n \in [0, \infty) \quad \textrm{a.s.} \quad \textrm{for every} \quad n = 1, 2,... \qquad\label{assu_L_U_bounded_one_side}
\end{eqnarray}
\end{remark}
The final step is to weaken condition (\ref{assu_L_U_monotone}) and let $L_n$ be a reverse time supermartingale and $U_n$ a reverse time submartingale with respect to $\mathcal{F}_{n, \infty}.$ Precisely, assume that for every $n = 1, 2, ...$ we have
\begin{eqnarray} \label{assu_L_super_MG}
\mathbb{E}\;(L_{n-1}\;|\; \mathcal{F}_{n,\infty}) \; = \; \mathbb{E}\;(L_{n-1}\;|\; \mathcal{F}_{n}) \; \leq \; L_n \; \textrm{ a.s.} &\textrm{and}& \\ \label{assu_U_sub_MG} \mathbb{E}\;(U_{n-1}\;|\; \mathcal{F}_{n,\infty}) \; = \; \mathbb{E}\;(U_{n-1}\;|\; \mathcal{F}_{n}) \; \geq \; U_n \; \textrm{ a.s.} 
\end{eqnarray}

Consider the following algorithm, that uses auxiliary random sequences $\tilde{L}_n$ and $\tilde{U}_n$ constructed online.
\begin{alg} \label{alg_4} ~\\
1. simulate  $G_0 \sim U(0,1);$ set $n = 1;$ set $L_0 \equiv \tilde{L}_0 \equiv 0$ and  $U_0 \equiv \tilde{U}_0 \equiv 1$ \\
2. obtain $L_n$ and $U_n$ given $\mathcal{F}_{0, n-1},$\\
3. compute $L_n^* = \mathbb{E}\; (L_{n-1}\; |\; \mathcal{F}_n)$ and $U_n^* = \mathbb{E}\; (U_{n-1}\; |\; \mathcal{F}_n).$\\
4. compute \begin{eqnarray}
\tilde{L}_n & = &\tilde{L}_{n-1} + \frac{L_n - L_n^*}{U_n^* - L_n^*} \left(\tilde{U}_{n-1} - \tilde{L}_{n-1} \right)\label{formula_tilde_L} \\ \tilde{U}_n & = & \tilde{U}_{n-1} - \frac{U_n^* - U_n}{U_n^* - L_n^*} \left(\tilde{U}_{n-1} - \tilde{L}_{n-1} \right) \label{formula_tilde_U}
\end{eqnarray}
5. if $G_0 \leq \tilde{L}_n$ set $C_s := 1;$\\
6. if $G_0 > \tilde{U}_n$ set $C_s := 0;$\\
7. if $\tilde{L}_n < G_0 \leq \tilde{U}_n$ set $n := n + 1$ and GOTO 2;\\
8. output $C_s.$
\end{alg}
\begin{thm} \label{thm_alg_4_valid}
Assume (\ref{assu_L_U_ineq}), (\ref{assu_L_U_in_01}), (\ref{assu_L_U_Expect}), (\ref{assu_L_super_MG}) and (\ref{assu_U_sub_MG}). Then Algorithm \ref{alg_4} outputs a valid $s-$coin. Moreover the probability that it needs $N>n$ iterations equals $u_n-l_n.$ \end{thm}
\begin{proof}
We show that $\tilde{L}$ and $\tilde{U}$ satisfy  (\ref{assu_L_U_ineq}), (\ref{assu_L_U_in_01}), (\ref{assu_L_U_Expect}) and (\ref{assu_L_U_monotone}) and hence Algorithm \ref{alg_4} is valid due to Lemma \ref{lemma_alg_3_valid}.

Conditions (\ref{assu_L_U_ineq}), (\ref{assu_L_U_in_01}) and (\ref{assu_L_U_monotone}) are straightforward due to construction of $\tilde{L}$ and $\tilde{U}$ and (\ref{assu_L_super_MG}), (\ref{assu_U_sub_MG}). 

To prove (\ref{assu_L_U_Expect}) we show that the construction in step 4 of Algorithm \ref{alg_4} preserves expectation, i.e. \begin{equation} \label{eqn_expect_for_alg_4} \mathbb{E} \; \tilde{L}_n = \mathbb{E} \; L_n = l_n \quad \textrm{ and } \quad  \mathbb{E} \; \tilde{U}_n = \mathbb{E} \; U_n = u_n. \end{equation} It is straightforward to check that (\ref{eqn_expect_for_alg_4}) holds for $n=1,2.$ Moreover note that $\tilde{U}_0 - \tilde{L}_0 = 1$ a.s., $U_1^* - L_1^* = 1$ a.s. and from (\ref{formula_tilde_L}) and (\ref{formula_tilde_U}) we have \begin{eqnarray} \nonumber  \tilde{U}_n - \tilde{L}_n & = & \left( \tilde{U}_{n-1} - \tilde{L}_{n-1} \right) \frac{U_n - L_n}{U_n^* - L_n^*} \qquad \textrm{and hence} \\ \tilde{L}_n & = & \tilde{L}_{n-1} +  \frac{L_n - L_n^*}{U_n^* - L_n^*} \frac{U_{n-1} - L_{n-1}}{U_{n-1}^* - L_{n-1}^*}  \cdots \frac{U_{2} - L_{2}}{U_{2}^* - L_{2}^*} \left(U_1 - L_1\right).\qquad \label{eqn_telescope} \end{eqnarray}
Now we compute $\mathbb{E}\; \tilde{L}_n$ by induction, conditioning (\ref{eqn_telescope}) subsequently on $\mathcal{F}_{2,\infty}, \dots, \mathcal{F}_{n, \infty}$ and using (\ref{assu_L_super_MG}) and (\ref{assu_U_sub_MG}). Calculation of $\mathbb{E}\; \tilde{U}_n$ is identical.
\begin{eqnarray}
\nonumber
\mathbb{E}\; \tilde{L}_n & = & \mathbb{E}\; \tilde{L}_{n-1} + \mathbb{E}\left(\mathbb{E} \left(\frac{L_n - L_n^*}{U_n^* - L_n^*} \frac{U_{n-1} - L_{n-1}}{U_{n-1}^* - L_{n-1}^*}  \cdots \frac{U_{2} - L_{2}}{U_{2}^* - L_{2}^*} \left(U_1 - L_1\right) \Big|
 \mathcal{F}_{2,\infty}\right)\right) \\
&=&  \mathbb{E}\; L_{n-1} + \mathbb{E}\left(\frac{L_n - L_n^*}{U_n^* - L_n^*} \frac{U_{n-1} - L_{n-1}}{U_{n-1}^* - L_{n-1}^*}  \cdots \frac{U_{2} - L_{2}}{U_{2}^* - L_{2}^*} \mathbb{E} \left(U_1 - L_1 \big|
 \mathcal{F}_{2,\infty}\right)\right) \nonumber \\ \nonumber 
&=&  \mathbb{E}\; L_{n-1} + \mathbb{E}\left(\frac{L_n - L_n^*}{U_n^* - L_n^*} \frac{U_{n-1} - L_{n-1}}{U_{n-1}^* - L_{n-1}^*}  \cdots \frac{U_{3} - L_{3}}{U_{3}^* - L_{3}^*}\left(U_2 - L_2\right)\right) = \cdots \\
& = & \mathbb{E}\; L_{n-1} + \mathbb{E}\; (L_n - L_n^*) \; =\;  \mathbb{E}\left(\mathbb{E}\left( L_{n-1} +  L_n - L_n^* \big| \mathcal{F}_{n, \infty} \right) \right) = \mathbb{E}\;L_n. \nonumber 
\end{eqnarray}
\end{proof}
\begin{remark}
All of the discussed algorithms are valid if $n$ takes values along an increasing
sequence $n_i \nearrow \infty.$
\end{remark}
Now let us link once again the algorithmic development of this Section with construction of unbiased estimators. Lemma \ref{lemma_alg_1_valid} together with Theorem \ref{thm_alg_4_valid} result in the following construction of
sequential unbiased estimators based on under- and overestimating reverse time
super- and submartingale sequences. The estimators are sequential in the sense that the amount of input needed to produce them is random. 
\begin{thm}
Suppose that for an unknown value of interest $s \in \mathbf{R},$ there exist a constant $M < \infty$ and random sequences $L_n$ and $U_n$ s.t.
\begin{eqnarray}
\mathbb{P}(L_n \leq U_n) = 1 & & \textrm{for every} \quad n = 1,2,... \label{assu_thm_L_U_ineq} \nonumber \\
\mathbb{P}(L_n \in [-M, M]) = 1 &\textrm{and}& \mathbb{P}(U_n \in [-M, M]) = 1 \quad \textrm{for every} \quad n = 1, 2,...\qquad \label{assu_thm_L_U_in_MM} \nonumber \\
\mathbb{E}\;L_n = l_n \nearrow s &\textrm{and}& \mathbb{E}\;U_n = u_n \searrow s \label{assu_thm_L_U_Expect} \nonumber \\
\label{assu_thm_L_super_MG}
\mathbb{E}\;(L_{n-1}\;|\; \mathcal{F}_{n,\infty}) & = & \mathbb{E}\;(L_{n-1}\;|\; \mathcal{F}_{n}) \; \leq \; L_n \; \textrm{ a.s.} \quad \textrm{and} \nonumber \\ \label{assu__thm_U_sub_MG} \mathbb{E}\;(U_{n-1}\;|\; \mathcal{F}_{n,\infty}) & = & \mathbb{E}\;(U_{n-1}\;|\; \mathcal{F}_{n}) \; \geq \; U_n \; \textrm{ a.s.} \nonumber
\end{eqnarray}
Then one can construct an unbiased estimator of $s.$
\end{thm}
\begin{proof}
After rescaling, one can use Algorithm 4 to sample events of probability $(M+s)/2M,$ which gives an unbiased estimator of $(M+s)/2M$ and consequently of $s.$
\end{proof}

\section{Application to the Bernoulli Factory Problem} \label{sec_NacuPeres}
Based on Section \ref{sec_theory}, we provide here a practical version of
the Nacu-Peres algorithm for simulating an $f(p)-$coin from a sequence of $p-$coins, where $f(p) =
\min\{2p, 1- 2\varepsilon\}.$ This is central to the general version of the Bernoulli factory problem, as \cite{NacuPeres} develops a calculus for collapsing simulation of a real analytic function, say $g,$ to simulation of $f(p) =
\min\{2p, 1- 2\varepsilon\}.$ Briefly, one takes a series expansion of $g$ and uses a composition of appropriate techniques (e.g. for simulating a sum or a difference of simulable functions). We refer to the original paper for details.   

In particular we prove Proposition
\ref{prop_mod_NacuPeres}, a general result, which is a minor
modification of Proposition 3 in \cite{NacuPeres}. However its proof,
different from the original one, links polynomial envelopes of $f$
with the framework of Section \ref{sec_theory} by identifying
terms. It results in an immediate application of Algorithm
\ref{alg_4}.

\begin{prop} \label{prop_mod_NacuPeres}
An algorithm that simulates a function $f$ on $\mathcal{P} \subseteq (0,1)$ exists if and only if for all $n\geq 1$ there exist polynomials $g_n(p)$ and $h_n(p)$ of the form
\begin{equation} \nonumber
g_n(p) = \sum_{k=0}^n \binom{n}{k}a(n,k)p^k(1-p)^{n-k} \quad \textrm{and} \quad h_n(p) = \sum_{k=0}^n \binom{n}{k}b(n,k)p^k(1-p)^{n-k},
\end{equation} s.t.
\begin{enumerate}
\item[(i)] $0 \leq a(n,k) \leq b(n,k) \leq 1,$
\item[(ii)] $\lim_{n\to \infty} g_n(p) = f(p) = \lim_{n\to \infty} h_n(p),$
\item[(iii)] For all $m <n,$ their coefficients satisfy \begin{eqnarray}\label{cond_coeff_a_b_MG} a(n,k)  \geq  \sum_{i=0}^k \frac{\binom{n-m}{k-i}\binom{m}{i}}{\binom{n}{k}} a(m,i), & &  b(n,k)  \leq  \sum_{i=0}^k \frac{\binom{n-m}{k-i}\binom{m}{i}}{\binom{n}{k}} b(m,i).\;\qquad \end{eqnarray}
\end{enumerate} 
\end{prop}
\begin{proof}
We skip the implication \textit{algorithm $\Rightarrow$ polynomials,} as it has been shown in \cite{NacuPeres}, and focus on proving \textit{polynomials $\Rightarrow$ algorithm} using framework of Section~\ref{sec_theory}. Let $X_1, X_2, \dots$ be a sequence of independent tosses of a $p-$coin. Define random sequences $\{L_n, U_n\}_{n\geq 1}$ as follows: if $\sum_{i=1}^n X_i = k,$ then let $L_n = a(n,k)$ and $U_n = b(n,k).$ In the rest of the proof we check that (\ref{assu_L_U_ineq}), (\ref{assu_L_U_in_01}), (\ref{assu_L_U_Expect}), (\ref{assu_L_super_MG}) and (\ref{assu_U_sub_MG}) hold for  $\{L_n, U_n\}_{n\geq 1}$ with $s=f(p).$ Thus executing Algorithm \ref{alg_4} with $\{L_n, U_n\}_{n\geq 1}$ yields a valid $f(p)-$coin.

Clearly (\ref{assu_L_U_ineq}) and (\ref{assu_L_U_in_01}) hold due to (i). For (\ref{assu_L_U_Expect}) note that $\mathbb{E}\;L_n = g_n(p) \nearrow f(p)$ and  $\mathbb{E}\;U_n = h_n(p) \searrow f(p).$ To obtain (\ref{assu_L_super_MG}) and (\ref{assu_U_sub_MG}) define the sequence of random variables $H_n$ to be the number of heads in $\{X_1, \dots, X_n\},$ i.e. $H_n = \sum_{i=1}^n X_i$ and let $\mathcal{G}_n = \sigma(H_n).$ Thus $L_n = a(n, H_n)$ and $U_n = b(n, H_n),$ hence $\mathcal{F}_n \subseteq \mathcal{G}_n$ and it is enough to check that $\mathbb{E}(L_m|\mathcal{G}_n) \leq L_n$ and $\mathbb{E}(U_m|\mathcal{G}_n) \geq U_n$ for $m<n.$ The distribution of $H_m$ given $H_n$ is hypergeometric and $$\mathbb{E}(L_m|\mathcal{G}_n)  = \mathbb{E}( a(m, H_m)| H_n) = \sum_{i=0}^{H_n} \frac{\binom{n-m}{H_n-i}\binom{m}{i}}{\binom{n}{H_n}}a(m,i) \leq a(n, H_n) = L_n.$$
Clearly the distribution of $H_m$ given $H_n$ is the same as the distribution of $H_m$ given $\{H_n, H_{n+1}, \dots \}.$ The argument for $U_n$ is identical.
\end{proof}
\begin{remark} \label{rem_factory}
In contrast to \cite{NacuPeres}, throughout this section we simulate $f$ in the weak sense, i.e. we use $U(0,1)$ as an auxiliary random variable. This is the natural approach in applications and also this is equivalent to strong simulability if $\mathcal{P} \subseteq (0,1),$ c.f. \cite{KeaneOBrien}.  
\end{remark}
To give a more complete view of the Bernoulli factory problem in the framework of Section \ref{sec_theory}, and before moving on to practical aspects of the problem, we show, as a corollary from Lemma \ref{lemma_alg_1_valid}, a result originally established in \cite{KeaneOBrien} and also provided in \cite{NacuPeres}, namely that generating $\min\{2p, 1\}-$coins from $p-$coins is not possible.

\begin{cor}\label{cor_p2p_does_not_exist} An algorithm that simulates $f(p) = 2p$ for $p \in \mathcal{P} = (0,1/2)$ does not exists.
\end{cor}
\begin{proof} We show that there does not exists an unbiased estimator of $2p$ for $p \in (0,1/2)$ that takes values in $[0,1]$ and we conclude the corollary from Lemma~\ref{lemma_alg_1_valid}. The idea of the proof is to show that such an estimator must take values smaller then $1/2$ with strictly positive probability independent of $p$ and then let $p \nearrow 1/2$ so that $2p \nearrow 1.$

Let $S$ be such an estimator and let $X_1, X_2, \dots$ be a sequence of $p-$coins. We allow $S$ to be sequential and use an auxiliary random variable $R_0$ independent of the $p-$coins. So $$S = S\big(\{X_1, X_2, \dots,\}, T, R_0\big) = S\big(\{X_1, X_2, \dots,X_T\}, R_0 \big),$$ where $T$ is a stopping time with respect to $\sigma\{\mathcal{F}_{1,n}, \mathcal{G}\},$ where $\{\mathcal{F}_{1,n}\}_{n\geq 1}$ is the filtration generated by $X_1, X_2, \dots$ and $\mathcal{G}$ is a $\sigma-$algebra independent of $\mathcal{F}_{1,n}$ and generated by $R_0.$ Clearly the joint distribution of $\left\{ \{X_1, X_2, \dots\}, T, R_0 \right\}$ depends on $p.$ We denote it by $\mathbb{P}_p$ and let $\mathbb{P}_{p|t}$ be the projection of $\mathbb{P}_p$ on $\{X_1, X_2, \dots, X_t\}.$ Now fix $p = 1/4.$ Since $2p = 1/2$ we have $\delta := \mathbb{P}_{1/4}(S \leq 1/2) > 0.$ Moreover there exists such an $t_0$ that $$\mathbb{P}_{1/4}(S \leq 1/2;\; T \leq t_0) \geq \delta/2.$$ Note that $\mathbb{P}_{p|t_0}$ is absolutely continuous with respect to $\mathbb{P}_{1/4|t_0}$ for all $p \in [1/4, 1/2)$ and
$$\inf_{p \in [1/4, 1/2)} \inf_{A \subseteq \{0,1\}^{t_0}} \frac{\mathbb{P}_{p|t_0}(A)}{\mathbb{P}_{1/4|t_0}(A)} \geq 2^{-t_0},$$
and consequently for every $p\in [1/4, 1/2)$ we have $$\mathbb{P}_{p}(S \leq 1/2)  \geq \mathbb{P}_{p}(S \leq 1/2;\; T \leq t_0) \geq 2^{-t_0} \mathbb{P}_{1/4}(S \leq 1/2;\; T \leq t_0) \geq  \delta 2^{-(t_0+1)}.$$ Now let $p \nearrow 1/2.$ This combined with $S \in [0,1]$ contradicts unbiasedness. 
\end{proof}

Given a function $f,$ finding polynomial envelopes satisfying properties required by Proposition~\ref{prop_mod_NacuPeres} is not easy. Section 3 of \cite{NacuPeres} provides explicit formulas for polynomial envelopes of $f(p) = \min\{2p, 1-2\varepsilon\}$ that satisfy conditions of Proposition \ref{prop_mod_NacuPeres}, precisely $a(n,k)$ and $b(n,k)$ satisfy (ii) and (iii) and one can easily compute $n_0 = n_0(\varepsilon)$ s.t. for $n \geq n_0$ condition (i) also holds, which is enough for the algorithm (however $n_0$ is substantial, e.g. $n_0(\varepsilon) = 32768$ for $\varepsilon = 0.1$ and it increases as $\varepsilon$ decreases). By Theorem \ref{thm_alg_4_valid} the probability that Algorithm \ref{alg_4} needs $N>n$ inputs equals $h_n(p) - g_n(p).$ The polynomials provided in \cite{NacuPeres} satisfy $h_n(p) - g_n(p) \leq C\rho^n$ for $p \in [0, 1/2 - 4\varepsilon]$ guaranteeing fast convergence, and $h_n(p) - g_n(p) \leq Dn^{-1/2}$ elsewhere. Using similar techniques one can establish polynomial envelopes s.t. $h_n(p) - g_n(p) \leq  C\rho^n$ for $p \in [0,1]\diagdown (1/2 - (2+c)\varepsilon, 1/2 - (2-c)\varepsilon).$ We do not pursue this here, however in applications it will be often essential to obtain polynomial approximations tailored for a specific problem and with desired properties. Moreover, we note that despite the fact that the techniques developed in \cite{NacuPeres} for simulating a real analytic $g$ exhibit exponentially decaying tails, they are often not practical. Nesting $k$ times the algorithm for $f(p)= \min\{2p, 1-2\varepsilon\}$ is very inefficient. One needs at least $n_0(\varepsilon)^k$ of original $p-$coins for a single output.

As mentioned earlier, a naive implementation of the Nacu-Peres
algorithm requires dealing with sets of $\{0,1\}$ strings of
exponential size. Other implementations with reduced algorithmic cost
are certainly possible with additional effort. However, our martingale
approach that uses Algorithm~\ref{alg_4} in the way indicated in the
proof of Proposition~\ref{prop_mod_NacuPeres}, avoids this problem
completely (a C-code for $f(p) = \min\{2p, 1-2\varepsilon\}$ is
available on request). 

Nevertheless, we note that for both algorithms, i.e. the original Nacu-Peres version and our martingale modification, the same number of original $p-$coins will be used for a single $f(p)-$coin output with $f(p) = \min\{2p, 1-2\varepsilon\}$ and consequently also for simulating any real analytic function using methodology of \cite{NacuPeres} Section~4. A significant improvement in terms of $p-$coins can be achieved only if the  monotone super/sub-martingales can be constructed directly and used along with Algorithm~\ref{alg_3}. This is discussed in the next subsection.

\subsection{Bernoulli Factory for alternating series expansions}
\label{sec:bern-alt}

In the following Proposition we describe an important class of functions for which an $f(p)-$coin can be simulated  by direct application of Algorithm~3.  

\begin{prop}\label{prop_alternating}
Let $f:[0,1] \to [0,1]$ have an alternating series expansion\begin{eqnarray}
\nonumber
f(p) & = & \sum_{k=0}^{\infty}(-1)^k a_k  p^k \qquad \quad \textrm{with} \quad 1 \geq a_0 \geq a_1 \geq \dots
\end{eqnarray}
Then an $f(p)-$coin can be simulated by Algorithm~3 and the probability that it needs $N > n$ iterations equals $a_n p^n.$  
\end{prop}
\begin{proof}
Let $X_1, X_2, \dots $ be a sequence of $p-$coins and define \begin{eqnarray}
U_0 & := & a_0 \qquad \qquad L_0 \;\; := \;\; 0,\nonumber \\
L_{n} &:=& \left\{\begin{array}{lll} U_{n-1} - a_{n} \prod_{k=1}^{n}X_k & \textrm{if} & n \; \textrm{ is odd,} \\ L_{n-1} & \textrm{if} & n \; \textrm{ is even,} \end{array} \right. \nonumber \\
U_{n} &:=& \left\{\begin{array}{lll} U_{n-1}  & \textrm{if} & n \; \textrm{ is odd,} \\ L_{n-1} + a_n  \prod_{k=1}^{n} X_k & \textrm{if} & n \; \textrm{ is even.} \end{array} \right. \nonumber
\end{eqnarray} 
Clearly (\ref{assu_L_U_ineq}), (\ref{assu_L_U_in_01}), (\ref{assu_L_U_monotone}) and (\ref{assu_L_U_Expect}) are satisfied with $s=f(p).$ Moreover, $$u_n - l_n = \mathbb{E}\,U_n - \mathbb{E}\,L_n =  a_n p^n \leq a_n.$$ Thus if $a_n \to 0,$ the algorithm converges for $p \in [0,1],$ otherwise for $p \in [0,1).$
\end{proof}

Next we illustrate the difference between application of Algorithm~\ref{alg_4} based on the Nacu-Peres approach and direct usage of Algorithm~\ref{alg_3} for simulating $f(p) = \exp(-ap), \; a<1.$ This function appears in applications discussed in Section~\ref{sec:exact}. Weak simulation is considered (c.f. Remark~\ref{rem_factory} and \cite{KeaneOBrien}), i.e. all normally available random variables are obtained directly, not from $p-$coins.

First consider sampling an $f(p)-$coin by collapsing the problem to doubling (i.e. sampling of $\min\{2p, 1- 2\varepsilon\}$) using techniques of \cite{NacuPeres} Section~4 and  Algorithm~\ref{alg_4}, we refer to the original paper for a complete description of the approach. The aim is to use as few doubling steps as possible, since doubling is expensive. Let $k := \in \{1,2\}$ be s.t. $2^k > e^a.$ 
We have \begin{eqnarray}\nonumber e^{-ap} &=& \sum_{n=0}^{\infty}\frac{a^{2n}}{(2n)!} p^{2n}-\sum_{n=0}^{\infty}\frac{a^{2n+1}}{(2n+1)!} p^{2n+1} \nonumber \\ &=& 2^k\bigg(\frac{e^a}{2^k}\Big( s_+(p) - s_-(p)\Big)\bigg),\label{eqn_ex_outer_doubling} 
\\ \textrm{where} && s_+(p) \;=\; \sum_{n=0}^{\infty}\frac{e^{-a}a^{2n}}{(2n)!} p^{2n}\qquad \textrm{and} \qquad  s_-(p) \; = \; \sum_{n=0}^{\infty}\frac{e^{-a}a^{2n+1}}{(2n+1)!}p^{2n+1}.\nonumber \end{eqnarray} 
First consider obtaining $\big( s_+(p) - s_-(p)\big)-$coins. This will be done by reversing $\big(1- s_+(p)) + s_-(p)\big)-$coins. Since 
\begin{eqnarray} \label{eqn_ex_doubling_inside} 1- s_+(p) + s_-(p) &=& 2\Big( \frac{1}{2}(1- s_+(p)) + \frac{1}{2}s_-(p)\Big),\end{eqnarray}
one feeds the doubling algorithm with $\big( \frac{1}{2}(1- s_+(p)) + \frac{1}{2}s_-(p)\big)-$coins obtained by tossing a fair coin first and using an $(1- s_+(p))-$coin or an $s_-(p)-$coin in the second step, depending on the outcome of the fair coin.  

We now describe sampling an $s_+(p)-$coin $C_{s_+(p)}.$ An
$s_-(p)-$coin can be obtained in a similar manner. Due to the specific
form of series expansion using a Poisson mixture is more efficient
then using the enforced geometric mixture suggested in the proof of
Proposition~16 \cite{NacuPeres}, details are below. Sample $N \sim
\textrm{Poiss}(a).$ If $N$ is even then generate iid $p-$coins $X_1,
\dots, X_N$ and declare $C_{s_+(p)}:= 1$ if $X_1 = \dots = X_N = 1.$
Otherwise, i.e. if $N$ is odd or $N$ is even and $\exists_{1\leq k\leq
  N}$ s.t. $X_k = 0,$ declare $C_{s_+(p)}:= 0.$ Finally, $\frac{e^a}{2^{k}}-$thinning should be applied to the $\big( s_+(p) - s_-(p)\big)-$coins and the doubling algorithm should be nested $k$ times on  $\frac{e^a}{2^{k}}\big( s_+(p) - s_-(p)\big)-$coins.

To approximate the total simulation effort let $K_p$ be the cost of obtaining the $p-$coin and assume $U(0,1)$ r.v's cost to be $O(1).$ Then the cost of a fair coin and the Poisson r.v. is also $O(1)$ (c.f. \cite{Devroye}). We assume $K_p \gg 1.$ Since $s_+(p) - s_-(p) = e^{-a(p+1)} \geq e^{-2a},$ 
\begin{eqnarray}\nonumber 
\frac{1}{2}(1- s_+(p)) + \frac{1}{2}s_-(p) &\leq & \frac{1}{2} - \frac{1}{2}e^{-2a},
\end{eqnarray}
and in (\ref{eqn_ex_doubling_inside}) one can take $\varepsilon_1 = e^{-2a}/5$ for the doubling algorithm for $\min\{2p, 1-2\varepsilon_1\}.$ Moreover to apply $k$ times the doubling scheme in (\ref{eqn_ex_outer_doubling}) we have to ensure that $\delta < e^{-ap} < 1-\delta$ (c.f. Proposition 17 of \cite{NacuPeres}) and therefore we have to restrict our considerations to the situation where a lower bound on $p,$ say $p_l$ is known. This implies that $\frac{e^a}{2}\big( s_+(p) - s_-(p)\big) = e^{-ap}/2 \leq e^{-ap_l}/2 = \frac{1}{2} - (\frac{1}{2}-e^{-ap_l}/2)$ and in the last iteration of the
 doubling scheme in (\ref{eqn_ex_outer_doubling}) one can take say $\varepsilon_2 =(\frac{1}{2}-e^{-ap_l}/2)/3.$ This yields a lower bound on the number of $p-$coins required before the algorithm can stop (c.f. the discussion in Section~\ref{sec_NacuPeres}), namely $n_0(\varepsilon_1)n_0(\varepsilon_2).$ If $k=2,$ the bound must be multiplied by $n_0(\varepsilon)$ with $\varepsilon$ used in the first iteration of the doubling scheme in (\ref{eqn_ex_outer_doubling}). Recall that e.g. $n_0(\varepsilon) = 32768$ for $\varepsilon = 0.1$ and it increases as $\varepsilon$ decreases. Therefore, when applying Algorithm~4 based on the Nacu-Peres approach, a conservative lower bound on the total simulation effort is $2^{30}K_p.$

On the other hand, for the exponential function we readily have an
alternating series expansion and can apply  Algorithm~3
directly by appealing to Proposition~\ref{prop_alternating}. Then,  we have \begin{eqnarray} \mathbb{P}(N \geq n) = \frac{(ap)^{n}}{n!},\nonumber \end{eqnarray}
where $N$ is the number of $p-$coins required. This implies $\mathbb{E}N \leq e,$ and the expected simulation effort is bounded from above by $eK_p$ and the bound holds uniformly for $a \in [0,1]$ and $p \in [0,1].$

\section{Application to the exact simulation of diffusions}
\label{sec:exact}

In this Section we derive the Exact Algorithm for diffusions
introduced by \cite{BeskosRobertsEA1} as a specific application of  the
Bernoulli factory for 
alternating series of Section \ref{sec:bern-alt}. 
 We are interested in
simulating $X_T$ which is the solution at time $T>0$ of the following
Stochastic Differential Equation (SDE):
\begin{equation}
\label{eq:SDE}
\mathrm{d}X_t = \alpha(X_t)\, 
\mathrm{d}t + \mathrm{d}W_t, 
\quad X_0 = x\in\mathbf{R}, \, t \in [0,T]
\end{equation}
driven by the Brownian motion 
$\{W_t\,;\,0\le t \le T\}$,
where the drift function $\alpha$ is assumed to 
satisfy the regularity conditions that 
guarantee the existence of a
weakly unique, global solution 
of (\ref{eq:SDE}), 
see ch.4 of \cite{kloeplat1995}.  Let $\Omega \equiv C([0,T],\mathbf{R})$ be the set
of continuous mappings from $[0,T]$ to $\mathbf{R}$ and $\omega$
be a typical element of $\Omega$. Consider the co-ordinate mappings
$B_t:\Omega\mapsto\mathbf{R}$, $t \in [0,T]$, such that for any $t$, $B_t
(\omega) = \omega(t)$ and the cylinder $\sigma$-algebra $\mathcal{C}
=\sigma (\{ B_t\,;\,0\leq t \leq T \})$. We denote by $W^x  = \{
W^x_t\,;\,0 \leq t \leq T\}$ the  Brownian motion started at $x \in
\mathbf{R}$, and by $W^{x,u}  = \{
W^{x,u}_t\,;\,0 \leq t \leq T\}$ the Brownian motion started at $x$
and finishing at $u \in
\mathbf{R}$ at time $T$; the latter is known as the Brownian bridge. 
We make the following assumptions for $\alpha$: 
\begin{enumerate}
\item The drift function $\alpha$ is 
   differentiable. 
\item The function
$h(u) = \exp \{ A(u) - (u-x)^2 / 2T \}$, 
$u\in\mathbf{R}$, for 
$A(u) = \int_
{\scriptscriptstyle{0}}^{\scriptscriptstyle{u}}
\alpha(y)\mathrm{d}y$,
is
integrable.
\item The function 
$(\alpha^2+\alpha^{'})/2$ is bounded below by $\ell>-\infty$, and
above by $r+\ell<\infty$.     
\end{enumerate}
Then, let us define 
\begin{equation}
  \label{eq:phi}
  \phi(u) = \frac{1}{r}[(\alpha^2+\alpha^{'})/2-\ell] \, \in [0,1]\,,
\end{equation}
$\mathbb{Q}$ be the probability measure induced by the solution $X$ of
(\ref{eq:SDE}) on $(\Omega,\mathcal{C})$, $\mathbb{W}$ the corresponding probability measure
for $W^x$, and $\mathbb{Z}$ be the probability 
measure defined as the following simple change of measure from $\mathbb{W}$:
$\mathrm{d}\mathbb{W}/ \mathrm{d}\mathbb{Z}( \omega) \propto \exp\{ -
A(B_T) \}$. Note that a stochastic process distributed according to
$\mathbb{Z}$ has similar dynamics to the Brownian motion,
with the exception of the distribution of the marginal distribution at
time $T$ which is biased according to $A$. Hence, we refer to 
this process as the biased Brownian motion. 
In particular,  the biased Brownian motion conditional on its value at
time $T$ has the same law as the corresponding Brownian bridge. 

The final steps of the mathematical developement entail resorting to
the Girsanov 
transformation of measures 
(see for instance ch.8 of \cite{okse:1998}) to obtain
$\mathrm{d}\mathbb{Q}/\mathrm{d}\mathbb{W}$;  applying an
integration-by-parts (possible by means of Assumption 1) to eliminate
the stochastic integral involved in 
the Radon-Nikodym derivative; and using the definition of $\mathbb{Z}$
to obtain that
\begin{equation}
\label{eq:final}
\frac { \mathrm{d} \mathbb{Q} } 
{ \mathrm{d} \mathbb{Z} } \, ( \omega )
\propto 
\exp\left\{ -rT \int_{0}^{T}T^{-1}
\phi
(B_{t})\mathrm{d}t\right \} \,\leq 1\,\qquad
  \mathbb{Z}-\textrm{a.s.}
\end{equation}
The details of this argument can be found in
\cite{BeskosRobertsEA1,BeskosPapaRobertsEA3}. 
By a standard rejection sampling principle, it follows that a path
$\omega$ generated according to $\mathbb{Z}$ and accepted with
probability \eqref{eq:final}, yields a draw from
$\mathbb{Q}$. 
Hence, the following algorithm yields  an exact sample
from the solution of 
\eqref{eq:SDE} at time $T$: \\

%
\noindent 1. simulate  $u \sim h$ \\
2. generate a $C_s$ coin where $s:=e^{-rT J}$, and $J:=\int_{0}^{T}T^{-1}
\phi
(W^{x,u}_{t})\mathrm{d}t$; \\ 
3. If $C_s=1$ output $u$ and STOP; \\
4. If $C_s=0$ GOTO 1. \\
%

\noindent Exploiting the Markov property, we can
assume from now on that $rT<1$. 
If  $T$ is such that
$rT>1$, then we can devise sub-intervals of length $\delta$ such that
$r\delta<1$ and apply the algorithm sequentially.

Clearly, the challenging part of the algorithm is Step 2,
since exact computation of $J$ is impossible due to the integration
over a Brownian bridge path.  On
the other hand, it is easy to generate $J$-coins: $C_J=\mathbb{I}(\psi
< \phi(W^{x,u}_\chi))$, where $\psi \sim U(0,1)$ and $\chi \sim
U(0,T)$ independent of the Brownian bridge $W^{x,u}$ and of
each other. Therefore, we deal with another instance of the problem
studied in this article: given $p$-coins how to generate $f(p)$-coins,
where here $f$ is the exponential function. This is precisely the
context of Section \ref{sec:bern-alt}, where the use of
Algorithm~\ref{alg_3} was advocated for efficient simulation.  As a
final remark, we note that exact simulation algorithms 
have been proposed in \cite{BeskosPapaRobertsEA2,BeskosPapaRobertsEA3}
for multivariate diffusions and  unbounded drit functionals. These
extensions involve decompositions of the  Brownian motion and  auxiliary Poisson
processes.

\section*{Acknowledgements} We would like to thank an anonymous referee for insightful suggestions that improved the presentation of the paper and  \c Serban Nacu and Yuval Peres for helpful comments. The first author is grateful to Wojciech Niemiro for a helpful discussion. 
The third author would like to acknowledge financial support by the Spanish government
through a ``Ramon y Cajal'' fellowship and the grant  MTM2008-06660
and 
the Berlin Mathematical School for hosting him as a visiting
Professor while preparing this manuscript. 
\end{document}